\documentclass[10.5pt, letterpaper]{article}

\usepackage[letterpaper, portrait, margin=1.5in]{geometry}

\usepackage[utf8]{inputenc}
\usepackage{setspace}
\usepackage{natbib}
\usepackage{graphicx}
\usepackage{hyperref}
\usepackage{amsfonts}
\usepackage{amsthm}

\usepackage{amsmath}
\numberwithin{table}{section}
\numberwithin{figure}{section}
\numberwithin{equation}{section}

\newtheorem{lemma}{Lemma}
\newtheorem{corollary}{Corollary}

\newtheorem*{corollary1}{Corollary 1}

\newtheorem*{lemma1}{Lemma 1}

\begin{document}

\title{Design-based conformal prediction}
\author{Jerzy Wieczorek\footnote{Jerzy Wieczorek, Department of Statistics, Colby College, Waterville, Maine, USA. E-mail: jawieczo@colby.edu.}}

\maketitle

\doublespacing

\begin{abstract}
Conformal prediction is an assumption-lean approach to generating distribution-free prediction intervals or sets, for nearly arbitrary predictive models, with guaranteed finite-sample coverage.
Conformal methods are an active research topic in statistics and machine learning, but only recently have they been extended to non-exchangeable data. In this paper, we invite survey methodologists to begin using and contributing to conformal methods. We introduce how conformal prediction can be applied to data from several common complex sample survey designs, under a framework of design-based inference for a finite population, and we point out gaps where survey methodologists could fruitfully apply their expertise. Our simulations empirically bear out the theoretical guarantees of finite-sample coverage, and our real-data example demonstrates how conformal prediction can be applied to complex sample survey data in practice.
\end{abstract}


\section{Introduction}\label{sec:Introduction}


\paragraph{What is conformal prediction?}
Conformal prediction, also called conformal inference, is a family of general-purpose approaches to constructing prediction intervals or prediction sets, which can be wrapped around almost any predictive modeling algorithm. Specifically, imagine that we have fit a function $\hat f_n$ to a set of training data $(X_i,Y_i)$ for $i=1,\ldots,n$ that allows us to make a point prediction for a new observation's response value $Y_{n+1}$ when $\hat f_n$ is evaluated at the covariate value $X_{n+1}$, and our goal is to report a level $1-\alpha$ prediction set for $Y_{n+1}$. Let us require only that all $n+1$ of the points $(X_i,Y_i)$ are an
exchangeable\footnote{A sequence of random variables is exchangeable when its joint probability distribution would not change if you permuted the random variables. More precisely, a sequence $X_1,X_2,\ldots$ is exchangeable if for each $n$ and each permutation $\pi$ of $\{1,\ldots,n\}$, the joint distribution of $(X_{\pi(1)}, \ldots, X_{\pi(n)})$ is just the same as the distribution of $(X_1,\ldots,X_n)$ \citep{durrett2019probability}. Examples of exchangeability include iid sequences as well as simple random sampling (with or without replacement).}
sample from some common distribution $P$, and that the algorithm used to fit $\hat f_n$ treats these points symmetrically (which rules out e.g.\ algorithms that give more weight to more-recent data over time).
If so, we can use standard conformal prediction methods to construct $\hat C_n$ which are either prediction intervals (for a regression problem) or prediction sets (for a classification problem), such that $\mathbb{P}[Y_{n+1} \in \hat C_n(X_{n+1})] \geq 1-\alpha$,
where the probability is taken over repeated sampling of all $n+1$ points. In Section~\ref{sec:Intuition} we describe two standard approaches to conformal prediction and describe the intuition behind how they work.

Although survey data analysis has traditionally focused on estimation and testing, predictive modeling with complex sample survey data is not uncommon. To name just a few examples: \cite{hong2010prediction} use the Second Longitudinal Study on Aging \citep{nchs2016lsoa} to fit a predictive model that can be used to monitor functional mobility status among the elderly. \cite{kshirsagar2017household} use 2015 Living Conditions Monitoring Survey data \citep{cso2015lcsm} to build a model to classify households by poverty level, at a range of poverty thresholds. \cite{krebs2019predicting} use Forest Inventory and Analysis data \citep{bechtold2005enhanced} to build random-forest models that can predict understory vegetation structure, intending for these models to be implemented in the Forest Vegetation Simulator \citep{crookston2005forest}. In each case, there is interest in using the model to make unit-level predictions, not just to estimate regression coefficients or population means. Yet, with the exception of \cite{hong2010prediction}, many such papers have not been able to provide prediction intervals or sets.

\paragraph{Why is this interesting to survey methodologists?}
For practitioners, survey data analysts who build predictive models---like those mentioned above---could apply conformal methods to provide design-based prediction intervals or sets (Section~\ref{sec:Methods}). Such prediction intervals or sets would have guaranteed coverage, even for novel machine learning algorithms.
Conformal methods may not add much value to older methods like linear regression, which already has well-known Gaussian-based prediction intervals and has already been adapted to account for complex survey designs. But for newer methods---such as prediction algorithms that have not yet been adapted to work optimally with survey designs, or non-probabilistic algorithms that do not come with ``built-in'' prediction intervals
such as the simulation models of \cite{leroy2021conformal}---we can apply conformal prediction methods that guarantee coverage based on the sampling design alone.
Those who collect and pre-process the survey data may also find uses for conformal methods, for instance in nonresponse prediction, imputation, or data cleaning (Section~\ref{sec:Extensions}).

More philosophically, the principles behind conformal methods align with complex survey sampling inference. [a] Conformal methods and many design-based methods provide exact, finite-sample coverage guarantees, not asymptotic or approximate guarantees.
[b] Unlike traditional model-based approaches (such as Gaussian-errors prediction intervals for regression), conformal prediction does not require assumptions about the distribution $P$---only exchangeable sampling. Likewise, design-based methods do not require assumptions about how the data values are distributed in the population---only knowledge of the sampling design.
Finally, [c] conformal guarantees hold even if the predictive model is not a good fit to $P$ (though of course the prediction intervals or sets may be larger then), and many model-assisted methods in design-based inference also have guarantees that hold whether or not the model is a good fit to the population.

Hence, we believe this is a good opportunity to cross-pollinate ideas between the survey methodology and machine learning research communities. The growth of conformal prediction indicates that machine learning practitioners are seeking methods that are guaranteed to work based only on the sampling design, not on the distribution of the data values. Survey methodologists, having the right expertise to meet that demand, may enjoy applying their skillset to this new challenge.

\paragraph{Why is this not already in use for survey sampling?}
Despite the apparent affinities, conformal methods have not previously been studied under a complex sample survey framework. Conformal prediction has been under development for several decades, initially by Vladimir Vovk and colleagues \citep{vovk2022algorithmic} and more recently by a wide range of statistical and machine learning researchers. Conformal methods are starting to be deployed in real-world settings, such as the Washington Post's 2020 presidential election tracker, which reported prediction intervals of the votes for each party that were updated in real time as voting districts slowly reported their results \citep{cherian2020washington}. Yet the requirement of exchangeability has made traditional conformal methods inapplicable for complex sampling designs other than simple random sampling (with or without replacement).

However, recently \cite{tibshirani2019conformal} extended conformal prediction to what they call ``weighted exchangeable'' sampling, while \cite{dunn2022distribution} derived several ``hierarchical'' conformal methods that can apply to cluster sampling. The present work builds on these two papers to begin providing conformal guarantees for complex sampling designs.

Briefly, instead of assuming that all $n+1$ data points are exchangeable, we will assume that the $n$ training points came from a known complex sampling design on a finite population, and the test point was selected uniformly at random from the same finite population.
The latter assumption is simply a way of phrasing the guarantee of marginal coverage across the finite population.
In Section~\ref{sec:Methods}, we summarize the complex survey settings in which such conformal prediction methods are currently known to have coverage guarantees.

\subsection{Related work}


Prediction intervals or tolerance regions have a long history in statistics; for a recent review, see \cite{tian2022methods}.
Section 13.3.3 of \cite{vovk2022algorithmic} connects conformal methods to much of this earlier work.
However, conformal prediction is currently a broad and very active research area, and we encourage interested readers to look into the following general resources: \cite{angelopoulos2022gentle} provide an introductory review article; Vovk and colleagues maintain a list of their own recent work at \url{http://alrw.net/}; \cite{manokhin2022awesome} maintains a frequently updated list of conformal papers, tutorials, and software; and COPA, the Symposium on Conformal and Probabilistic Prediction with Applications, is an annual conference on conformal prediction and its extensions: \url{https://copa-conference.com/}.

For dependent data, our paper builds largely on the work of \cite{dunn2022distribution} and \cite{tibshirani2019conformal}.
\cite{fong2021conformal} give a distinct, Bayesian approach to conformal prediction for hierarchical data.
\cite{barber2022conformal} relax the usual conformal-prediction requirement that the fitted prediction model must treat data points symmetrically, and \cite{fannjiang2022conformal} address ``feedback covariate shift,'' where the choice of which test data to sample depends on results from the training data, which could be useful in addressing adaptive sampling for surveys \citep{thompson1997adaptive}.
Conformal methods for nonstationarity over time \citep{chernozhukov2018exact, oliveira2022split}, including distribution drift \citep{gibbs2021adaptive}, would be useful in developing conformal methods for panel surveys. Recently, \cite{lunde2023validity} developed conformal methods for network data under non-uniform sampling, with implications for respondent-driven sampling in surveys.

Few papers have applied conformal methods to survey samples. \cite{bersson2022optimal} do use conformal prediction for small area estimation problems which rely on survey data. However, they work in a model-based framework and assume that sampled units are exchangeable within each small area, rather than focusing on design-based inference under a complex sampling design.
\cite{romano2019conformalized} and several followup papers \citep{sesia2020comparison, sesia2021conformal, feldman2021improving, bai2022efficient} illustrate conformal prediction methods using unit-level records from a complex survey: the Medical Expenditure Panel Survey \citep{ahrq2017meps}. However, they appear to ignore the documented sampling design and treat it as exchangeable.

Apart from conformal methods, there has been increasing interest in applying machine learning methods to complex survey data. \cite{dagdoug2022model} summarize the current state of the art for using survey data with many machine learning algorithms: penalized regression \citep{mcconville2017model}, k nearest neighbors \citep{yang2019nearest}, random forests \citep{dagdoug2021model}, and others. However, along with \cite{sande2021design}, they focus on estimation of means and totals, rather than on predicting individual response values. Even when such algorithms have been adapted to make individual predictions that account for the survey design, many still do not offer procedures for creating prediction intervals, which design-based conformal methods could provide.

\subsection{Our contributions}

We introduce survey methodologists to key definitions and intutions behind standard conformal methods under exchangeable sampling as well as their extensions to complex sampling (Section~\ref{sec:Intuition}).
Furthermore, we derive exact, finite-sample, design-based coverage guarantees for applying conformal inference to data from several fundamental sampling designs, building on results from the weighted exchangeable or hierarchical frameworks (Section~\ref{sec:Methods}).

Next, we illustrate our design-based guarantees through simulations and a real-data example (Section~\ref{sec:Examples}). We show that results can differ on real data depending on whether or not conformal methods account for the survey design, and we show how conformal methods are affected by different sampling designs.
We then discuss several practical considerations and future challenges to be addressed in applying conformal methods to survey data (Section~\ref{sec:Extensions}). We also suggest ways that conformal methods might be useful in conducting surveys.
Finally, we invite survey researchers to contribute to the literature on conformal methods (Section~\ref{sec:Conclusion}). We anticipate that advances from the design-based perspective will turn out to be useful to the general community of conformal researchers.

\section{Introduction to conformal inference}\label{sec:Intuition}

\subsection{Definitions}

Because one of our main contributions is a corollary of results in \cite{tibshirani2019conformal}, we restate without proof some of their notation, definitions, and results here and in Section~\ref{sec:Methods}.

Let $\mathrm{Quantile}(\beta; F)$ denote the level $\beta$ quantile of distribution $F$, so that for $Y\sim F$,
\[\mathrm{Quantile}(\beta;F) = \inf \{y: \mathbb{P}\{Y \leq y\} \geq \beta\}. \]
We allow for distributions on the augmented real line, $\mathbb{R} \cup \{\infty\}$. We use $v_{1{:}n}=\{v_1,\ldots,v_n\}$ to denote a multiset, meaning that it is unordered and can allow the same element to appear several times. We use $\delta_a$ to denote a point mass at the value $a$. If $v_{1{:}n}$ is an exchangeable sample, its empirical probability measure is $n^{-1}\sum_{i=1}^n \delta_{v_i}$, and the level $\beta$ quantile of its empirical distribution is $\mathrm{Quantile}(\beta; v_{1{:}n})$ which is the $\lceil \beta n \rceil$ smallest value in $v_{1{:}n}$, where $\lceil \cdot \rceil$ is the ceiling function.

Now we can state the general-purpose ``quantile lemma'' that forms the foundation for conformal inference in the exchangeable setting. For instance, if we apply this lemma directly to data from an exchangeable sample of scalar random variables $V_1,\ldots,V_n$, we obtain a level $\beta$ one-sided prediction interval for a new observation $V_{n+1}$.

\begin{lemma1}[\cite{tibshirani2019conformal}]
If $V_1,\ldots, V_{n+1}$ are exchangeable random variables, then
for any $\beta\in(0,1)$, we have
\[\mathbb{P}\left\{ V_{n+1} \leq \mathrm{Quantile}\left(
\beta; V_{1{:}n} \cup {\infty}
\right) \right\} \geq \beta.\]
Furthermore, if ties between $V_1,\ldots, V_{n+1}$ occur with probability zero, then the above probability is upper bounded by $\beta + 1/(n+1)$.
\end{lemma1}

In order to use this lemma to do conformal prediction for nontrivial regression or classification problems, we must also choose a \emph{score function} $\mathcal{S}$, which takes the following arguments: a point $(x,y)$, and a multiset $Z$ (meaning that $\mathcal{S}$ must treat the points in $Z$ as unordered). The points in $Z$ will typically be the $n$ observations $Z_i=(X_i,Y_i)$ for $i=1,\ldots,n$, possibly along with one additional hypothetical observation or test case; $(x,y)$ is typically one of the points in $Z$. $\mathcal{S}$ should return a real value, such that lower values indicate that $(x,y)$ ``conforms'' to $Z$ better. Finally, use $\mathcal S$ to define \emph{nonconformity scores}
\begin{equation}\label{eqn:Tibs3}
V_i^{(x,y)} = \mathcal{S}\left(Z_i,Z_{1{:}n} \cup \{(x,y)\}\right),\ i=1,\ldots,n,
\ \mathrm{and} \
V_{n+1}^{(x,y)} = \mathcal{S}\left((x,y), Z_{1{:}n}\cup\{(x,y)\}\right).
\end{equation}

The reason we require $Z$ to be an unordered multiset is to ensure that if the observations $(X_i,Y_i)$ are exchangeable, then the scores $V_i^{(x,y)}$ will be too.

For instance, in a regression context, we might choose
$\mathcal{S}\left((x,y),Z\right) = |y-\hat{f}(x)|$ where $\hat f$ is some regression function fitted using all of $Z$, including $(x,y)$. If the absolute-residual nonconformity score $V_{n+1}^{(x,y)}$ is small relative to all of the scores $V_i^{(x,y)}$, this suggests that $(x,y)$ conforms well to the overall trend in $Z$.

In a classification setting, we might choose $\mathcal{S}\left((x,y),Z\right) = 1-\hat{f}(x)_y$ where $\hat f$ is the probability of class $y$ estimated by some classification function fitted using all of $Z$, again including $(x,y)$. A small value of this nonconformity score $V_{n+1}^{(x,y)}$ suggests that $(x,y)$ conforms well to $Z$, in the following sense: This observation with covariates $x$ has a high estimated probability of being in class $y$ based on trends in $Z$, and it does indeed belong to class $y$.

We will also use following common abbreviations: SRS (simple random sampling), WR (with replacement), WOR (without replacement), PPS (probability proportional to size sampling). $N$ denotes a population size and $n$ typically denotes a sample size, except for the ``split conformal'' methods (defined in Section~\ref{sec:SplitFull}) which split the data into a ``proper training set'' of size $m$ and a ``calibration set'' of size $n$.

\subsection{Intuition under exchangeable sampling}\label{sec:IntuitionIid}

\paragraph{Quantile lemma:} The lower bound in Lemma~1 has a simple rationale. First, let $\hat q_{n+1}$ be the $\lceil \beta (n+1) \rceil$ smallest value in $V_{1{:}{(n+1)}}$, and let $\hat q_{conf}$ be the $\lceil \beta (n+1) \rceil$ smallest value in $V_{1{:}n}\cup\infty$.
By exchangeability, the rank of $V_{n+1}$ among all the $V_i$ is uniformly distributed over $\{1,\ldots,n+1\}$, so $V_{n+1}$ is at or below $\hat q_{n+1}$ with probability exactly $\beta$.
Then since $\hat q_{conf} \geq \hat q_{n+1}$, $V_{n+1}$ is at or below $\hat q_{conf}$ with probability \emph{at least} $\beta$. Note that these probabilities are marginal (over exchangeable sampling of all the $V_1,\ldots,V_{n+1}$ together)---not conditional on the first $n$ observations nor conditional on the test observation.

This is an exact finite-sample result that only relies on exchangeability. By contrast, if we used instead $\hat q_n$, the $\lceil \beta n \rceil$ smallest value in $V_{1{:}n}$, then it may be an asymptotically good estimate of the population quantile, but it would require stronger conditions on the data distribution. Even then, the probability that $V_{n+1}$ is at or below $\hat q_n$ would be only \emph{approximately} $\beta$.

To see why $\hat q_{conf} \geq \hat q_{n+1}$, imagine an empirical cumulative distribution function (eCDF) plot of just the first $n$ datapoints, with step heights of $1/n$ at each observed value.
(See the top left subplot of Figure~\ref{fig:eCDFs}.)
If we add one more datapoint, the step heights of the eCDF will be $1/(n+1)$ instead, and the lower $\beta$ quantile will be $\hat q_{n+1}$ as above. It may be larger or smaller than the lower $\beta$ quantile of the first $n$ datapoints, depending on how large the last datapoint is. The lower $\beta$ quantile is largest if we choose a $(n+1)^{th}$ value that is larger than any of the first $n$ observed values, say at $\infty$. That is because if the added datapoint is larger than any of the others, it pushes the rest of the eCDF down, so the horizontal line with $y$-intercept of $\beta$ crosses the eCDF at an $x$ value further to the right. In this most extreme case, the lower $\beta$ quantile is $\hat q_{conf}$ as above, and therefore $\hat q_{conf} \geq \hat q_{n+1}$.
(See the bottom left subplot of Figure~\ref{fig:eCDFs}.)
By contrast, if instead of $\infty$ we had chosen to add a value smaller than some of the first $n$ observations, it would have pushed part of the eCDF upward and may have caused the lower $\beta$ quantile to become smaller than $\hat q_{n+1}$.

\begin{figure}[h!]
\includegraphics[width=0.49\textwidth]{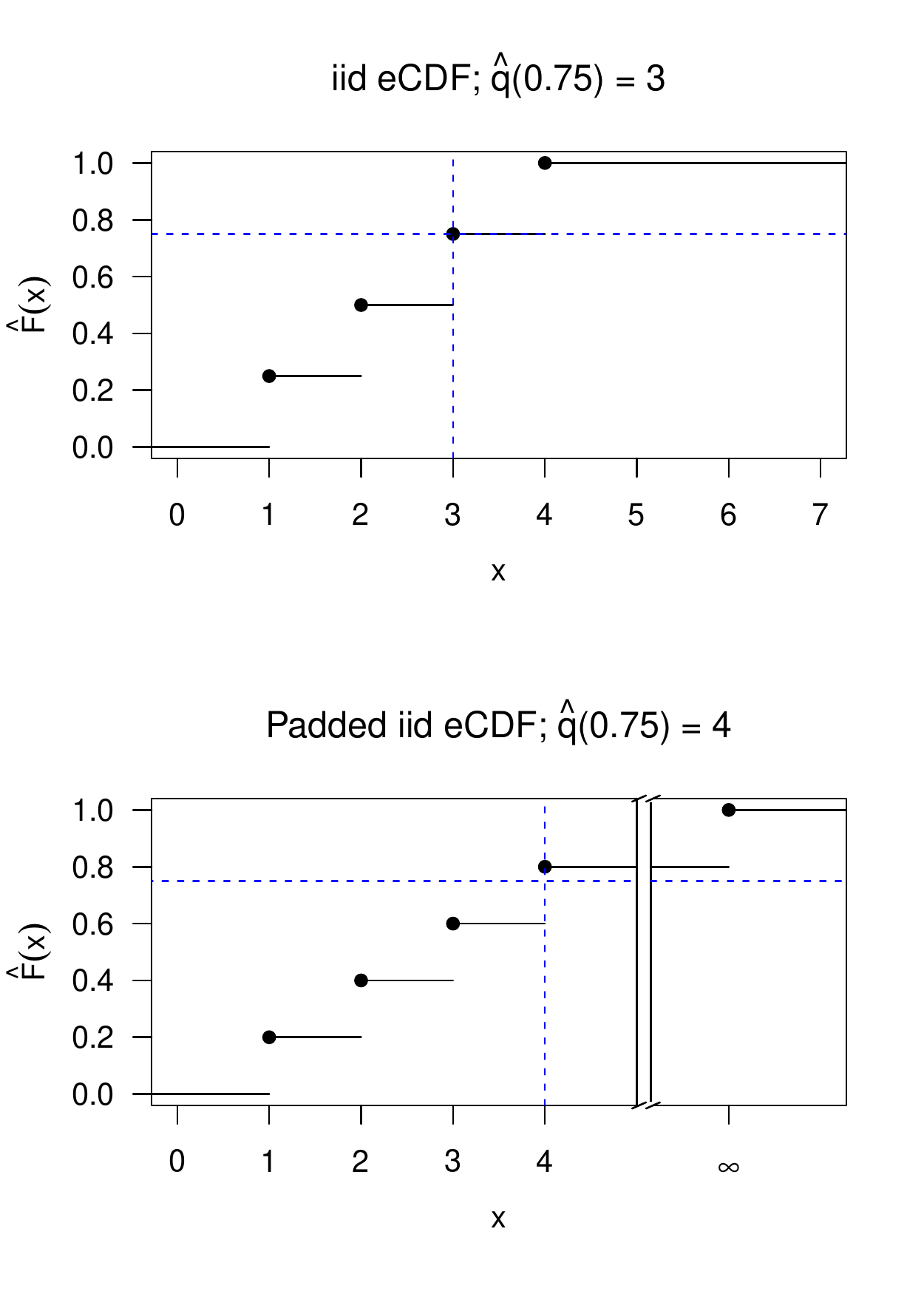}
\includegraphics[width=0.49\textwidth]{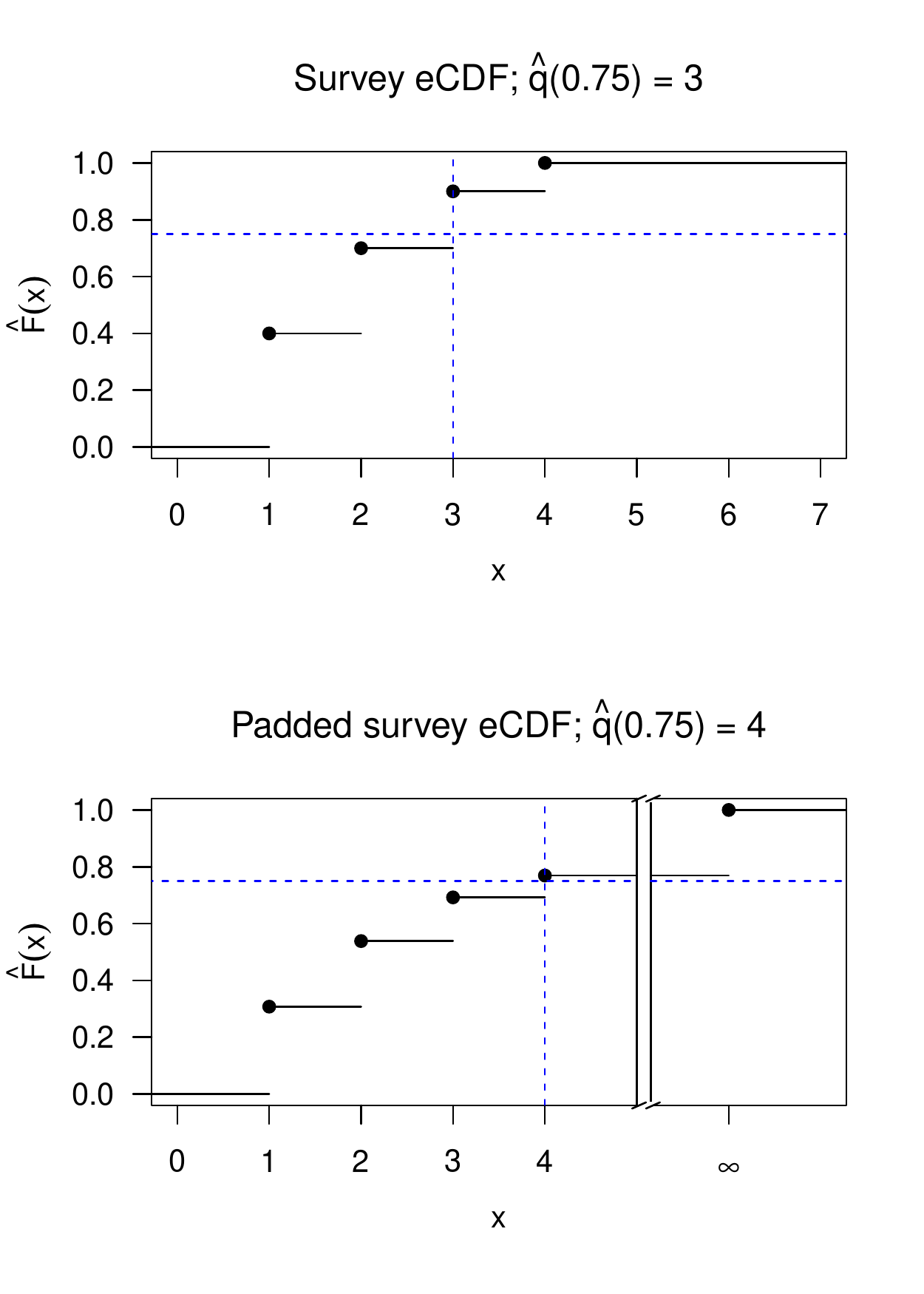}
\centering
\caption{Top left: eCDF for an iid sample with values \{1,2,3,4\}. At each value there is a vertical jump of $\frac{1}{n}$=0.25. The 75th percentile is at 3. \\ Bottom left: same eCDF but padded with an extra value at $\infty$. Now at each value there is a vertical jump of $\frac{1}{n+1}$=0.20. The 75th percentile is between 3 and 4, so we round up to 4. \\ Top right: eCDF for a survey sample with values \{1,2,3,4\} and corresponding survey weights \{4,3,2,1\}. At each value there is a vertical jump proportional to its survey weight. The 75th percentile is between 2 and 3, so we round up to 3. \\ Bottom right: same eCDF but padded with an extra value at $\infty$, corresponding to a not-yet-observed unit in the population that would have a survey weight of 3 if it were to be sampled. The vertical jumps at each value are still proportional to the survey weights, but now rescaled to make room for the weight of the extra value at $\infty$. The 75th percentile is between 3 and 4, so we round up to 4.}
\label{fig:eCDFs}
\end{figure}

Put another way, the probability that \{we take an exchangeable sample of size $n+1$ in which the last observation is greater than the $\lceil (1-\alpha)(n+1) \rceil$ smallest of the first $n$ observations\} is at most $\alpha$. This is a strict inequality, not approximate or asymptotic.

\paragraph{Regression prediction intervals:} Next, we can apply this quantile lemma in a regression setting. (In Section~\ref{sec:Class} we describe an approach for classification problems.) Imagine that we are given a fixed regression function $\hat f$ for making predictions of some real-valued random variable $Y$ from some covariate vectors $X$ for data from this population. (For the moment, assume $\hat f$ is not data-dependent; perhaps it was chosen a priori.) Also imagine that we have some residuals from applying $\hat f$ to a ``calibration set''\footnote{We recognize that the term ``calibration'' traditionally has a completely different meaning in survey sampling. In this paper, we will use ``calibration'' only in this conformal inference sense.} of $n$ different data points $(X_1,Y_1),\ldots,(X_n,Y_n)$.
Finally, we have one other covariate vector $X_{n+1}$, drawn exchangeably from the same population, but we do not know its $Y$ value. We want to get a level $(1-\alpha)$ prediction interval for $Y_{n+1}$ at this test-case $X_{n+1}$.

We can simply apply the ``quantile lemma'' logic to the $n$ absolute residuals from the calibration set. These absolute residuals are our nonconformity scores.
If we let $\hat q$ be the $\lceil(1-\alpha)(n+1)\rceil$ smallest absolute residual,
then $\hat C(X_{n+1}) = \hat f (X_{n+1}) \pm \hat q$ is a prediction interval for $Y_{n+1}$ with
guaranteed coverage of at least $1-\alpha$.

\paragraph{Marginal vs.\ conditional coverage:} This coverage is marginal across all samples of size $(n+1)$, with a size-$n$ calibration set plus one new test case. It is not conditional on the specific test case $X_{n+1}$ we chose; it assumes that both the calibration set and the test case together are exchangeable. Nonetheless, it is an exact finite-sample result. \cite{barber2021limits} show that although no conformal method could generally guarantee ``exact conditional coverage'' (conditioning on the exact value of $X$), certain relaxed versions of conditional coverage are achievable. \cite{angelopoulos2022gentle} review approaches to assessing and controlling various forms of conditional coverage, noting that marginal coverage alone may be insufficient, e.g.\ if it happens to be achieved by low coverage in a rare but important subpopulation and high coverage elsewhere. In Appendix~\ref{sec:Adaptive} we review links between conditional coverage and the related idea of ``adaptive'' prediction regions.

\subsection{Intuition under complex sampling}\label{sec:IntuitionWtd}

When the data are not exchangeable,
\cite{tibshirani2019conformal} have extended conformal prediction to a setting called ``covariate shift,'' in which the distribution of $X$ is differerent for the training and/or calibration sets than for the test set, but the conditional distribution of $Y|X$ remains the same.
They define a condition called ``weighted exchangeability'' and show how to make conformal inference work for such data. In the present paper, we will show that certain classic finite-population sampling designs can be treated as a special case of covariate shift, and therefore the results of \cite{tibshirani2019conformal} apply.

Specifically, imagine we sample $n$ cases with replacement from a finite population with known but unequal sampling probabilities; for instance, we may be using PPS sampling. Assume also that we take a SRS of just one test case from the full finite population. Now, we need a ``survey weighted quantile lemma'' that tells us how to find an adjusted quantile $\hat q$ of the complex sample, such that a test case nonconformity score is no larger than $\hat q$ with probability at least $1-\alpha$. We will plug the $n$ complex sample observations from our calibration set into a regression function $\hat f$ and find $\hat q$ for the absolute residuals. Then the probability of our test case $(X_{n+1},Y_{n+1})$ having a larger absolute residual is at most $\alpha$, and so $\hat f(X_{n+1}) \pm \hat q$ is a level $1-\alpha$ prediction interval for $Y_{n+1}$.

In Section~\ref{sec:Methods} we will show that in this unequal-probabilities setting, it is enough to mimic the exchangeable setting, except that instead of using an equal-weights eCDF to get the quantiles, we use a survey-weighted eCDF; see e.g.\ Section 5.11 of \cite{sarndal1992model}. The additional observation at $\infty$ will simply be assigned the known sampling weight that the test case $(X_{n+1},Y_{n+1})$ would have had, under the complex sampling design used to sample the first $n$ units.
See the right half of Figure~\ref{fig:eCDFs} for an illustration.

We acknowledge that there is a long-running debate in the survey sampling literature about whether and how to use sampling weights for model fitting and inference \citep{fienberg2010relevance, lumley2017fitting}.
Our goal in the present paper is not to take a stance in this debate, but simply to show how conformal inference can be applied \emph{if} the data analyst is taking a design-based perspective. In a model-based analysis, if the design features can justifiably be ignored, standard conformal inference methods may be used.



\subsection{Split vs.\ full conformal}\label{sec:SplitFull}

Above, for simplicity we have assumed that a $\hat f$ has been provided for us. More typically, we will need to fit $\hat f$ using data from the sample at hand. In the ``split conformal'' approach \citep{lei2018distribution}, also called ``inductive conformal'' \citep{papadopoulos2002inductive}, we start with an exchangeable dataset of $m+n$ observations, and we split them at random into a ``proper training set'' of size $m$ to be used for training $\hat f$, plus a calibration set of size $n$ as described above. In this situation, coverage is still marginal over the calibration set plus one test case, but now it is conditional on the proper training set.

An optimal sample splitting ratio $m/n$ for split conformal is not known. Larger $m/n$ should lead to better estimates of $\hat f$, and therefore shorter conformal prediction interval lengths on average. On the other hand, larger $m/n$ also leads to a smaller calibration set, and therefore more-variable conformal prediction interval lengths.

If we do not wish to lose statistical efficiency by splitting our data, we can use a more computationally-intensive ``full conformal'' approach \citep{vovk2022algorithmic}. In that approach, we no longer need a separate calibration set, and we let $n$ denote the total number of our complete-data training cases.
We also have one test case with only the covariates $X_{n+1}$ known. Then we repeat the following process for many $y$ values:

\begin{itemize}
  \item Choose a hypothetical response value $y \in \mathbb{R}$.
  \item Fit $\hat f_y$ to an augmented dataset in which we pretend this $y$-value is correct: \\ $(X_1,Y_1),\ldots,(X_n,Y_n),(X_{n+1},y)$. Find the $n+1$ nonconformity scores, e.g.\ the absolute residuals: $R_{y,i} = |Y_i-\hat f_y(X_i)|$ for $i=1,\ldots,n$ and  $R_{y,n+1}=|y-\hat f_y(X_{n+1})|$. Find their $1-\alpha$ quantile: $\hat q_y$ is the $\lceil (1-\alpha)(n+1) \rceil$ smallest value of $R_{y,i}$ for $i=1,\ldots,n+1$.
  \item If $R_{y,n+1} \leq \hat q_y$, we say $y$ ``conforms'' to the rest of the data, and we add $y$ to our prediction interval.
\end{itemize}

In terms of notation, refer back to \eqref{eqn:Tibs3}. Full conformal nonconformity scores are calculated using $Z = Z_{1{:}n} \cup \{(X_{n+1},y)\}$ and refitting a new prediction model $\hat f_y$ for each new $y$ or $X_{n+1}$. By contrast, split conformal nonconformity scores use a fixed $\hat f$ conditional on the proper training set. Each of the first $n$ nonconformity scores is calculated only using that data point $(X_i,Y_i)$ and the fixed $\hat f$, and there is no need to calculate a $(n+1)^{th}$ nonconformity score. Under the notation in \eqref{eqn:Tibs3}, we allow $Z$ to be ignored and $V_{n+1}^{(x,y)}$ to be undefined for split conformal.

Compared to split conformal, the full conformal procedure reduces the variance in the reported prediction interval endpoints. On the other hand, full conformal is far more computationally intensive than split conformal, as $\hat f_y$ has to be refit for every candidate $y$ at a given test point $X_{n+1}$, and the entire interval needs to be refit for every test point.

In empirical comparisons, \cite{lei2018distribution} found that split conformal and full conformal often produce very similar prediction intervals. They recommend using split conformal, on the grounds that it is faster to compute with little loss of efficiency, although full conformal avoids randomness in the data split.


\subsection{Classification problems and prediction sets}\label{sec:Class}

For brevity, we focus on regression prediction intervals, but classification prediction sets are another common use case.
For instance, consider split conformal inference for a multi-class probabilistic classifier. Examples range from simple logistic regression to deep neural networks. First, fit the classifier to the proper training set to get a function $\hat f(x)_y$, whose outputs are the estimated probabilitity of class $y$ for an input $x$. Then, using the calibration set, each calibration ``residual'' or nonconformity score $V_i$ is calculated as $1-\hat{f}(X_i)_{Y_i}$ for each calibration observation $i=1,\ldots,n$.
Next, find the corrected $1-\alpha$ quantile of these probabilities: let $\hat q$ be the $\lceil (1-\alpha)(n+1)\rceil$ smallest value of $V_i$ for $i=1,\ldots,n$.
Finally, for a new test case $X_{n+1}$, find its estimated probability for each class $y$, and choose all the classes whose nonconformity scores are below the corrected quantile: $\hat{C}(X_{n+1}) = \{y: 1-\hat f(X_{n+1})_y \leq \hat q \}$.
\cite{romano2020classification} and \cite{angelopoulos2022gentle} discuss refinements to this approach.

\section{Methods}\label{sec:Methods}

For the reader's convenience, we restate several of the key results from \cite{tibshirani2019conformal} below. Next we prove that we can apply these results to unequal-probability sampling with replacement. We follow by discussing methods for sampling without replacement, cluster sampling based on \cite{dunn2022distribution}, stratified sampling, and post-stratification.

Most of the methods below apply both to full and split conformal. However, for split conformal,
we recommend using the design-based approach of \cite{wieczorek2022kfold} to split survey data into a proper training set and a calibration set.
This way, both the proper training set and the calibration set will mimic the original sampling design, which ensures that the methods below are safe to apply to your calibration set. By contrast, a simple random split can cause the calibration set to have different properties than a clustered or stratified sampling design.

\subsection{Previous results for covariate shift  \citep{tibshirani2019conformal}}\label{sec:CovShift}

\cite{tibshirani2019conformal} define the ``covariate shift'' setting as follows:
\begin{align}\label{eqn:Tibs5}
(X_i, Y_i) &\stackrel{iid}{\sim} P = P_X \times P_{Y|X},\ i=1,\ldots,n, \nonumber \\
(X_{n+1}, Y_{n+1}) &\sim \tilde{P} = \tilde{P}_X \times P_{Y|X},\ \textrm{independently.}
\end{align}
Note that the conditional distribution of $Y|X$ remains the same as the marginal distribution of $X$ changes.

Assuming that $P_X$ and $\tilde{P}_X$ are known, we can define likelihood ratio weight functions $w = \mathrm{d}\tilde{P}_X / \mathrm{d}P_X$.
We use these to define a second set of weights:
\begin{equation}\label{eqn:Tibs6}
p_i^w(x) = \frac{w(X_i)}{\sum_{j=1}^n w(X_j) + w(x)},\ i=1,\ldots,n,
\quad \mathrm{and} \quad
p_{n+1}^w(x) = \frac{w(x)}{\sum_{j=1}^n w(X_j) + w(x)}.
\end{equation}

In this setting, we can state a weighted, nonexchangeable counterpart to Lemma~1.
Although \cite{tibshirani2019conformal} give and prove a more general version, here we only state a version tailored to the covariate shift setting.

\setcounter{lemma}{1}
\begin{lemma}\label{lem:weightedquantile}
Assume data from the model~\eqref{eqn:Tibs5}.
Assume $\tilde{P}_X$ is absolutely continuous with respect to ${P}_X$, and denote $w = \mathrm{d}\tilde{P}_X/\mathrm{d}P_X$.
For any $\beta\in(0,1)$,
\[\mathbb{P}\left\{ V_{n+1} \leq \mathrm{Quantile}\left(
\beta; \sum_{i=1}^n p_i^w(x) \delta_{V_i} + p_{n+1}^w(x) \delta_\infty
\right) \right\} \geq \beta,\]
where $V_i^{(x,y)}$, $i=1,\ldots,n+1$ are as defined in \eqref{eqn:Tibs3}, and $p_i^w$, $i=1,\ldots,n+1$ are as defined in \eqref{eqn:Tibs6}.
\end{lemma}
\begin{proof}
Apply Lemma~3 of \cite{tibshirani2019conformal} in the covariate shift setting of \eqref{eqn:Tibs5}.
\end{proof}

This weighted quantile lemma allows conformal inference in the covariate shift setting.

\begin{corollary1}[\cite{tibshirani2019conformal}]
Assume data from the model~\eqref{eqn:Tibs5}.
Assume $\tilde{P}_X$ is absolutely continuous with respect to ${P}_X$, and denote $w = \mathrm{d}\tilde{P}_X/\mathrm{d}P_X$.
For any score function $\mathcal{S}$, and any $\alpha\in(0,1)$, define for $x \in \mathbb{R}^d$,
\begin{equation*}\label{eqn:Tibs7}
\hat{C}_n(x) = \biggl\{ y\in\mathbb{R}:V_{n+1}^{(x,y)} \leq \mathrm{Quantile} \biggl( 1-\alpha; \sum_{i=1}^n p_i^w(x) \delta_{V_i^{(x,y)}} + p_{n+1}^w(x) \delta_\infty \biggr) \biggr\},
\end{equation*}
where $V_i^{(x,y)}$, $i=1,\ldots,n+1$ are as defined in \eqref{eqn:Tibs3}, and $p_i^w$, $i=1,\ldots,n+1$ are as defined in \eqref{eqn:Tibs6}.
Then $\hat{C}_n$ satisfies $\mathbb{P}\left\{ Y_{n+1} \in \hat{C}_n(X_{n+1})\right\} \geq 1-\alpha$.
\end{corollary1}

This is the weighted ``full conformal'' approach. For the ``split conformal'' approach, we restate part of Section~A.3 from the supplement to \cite{tibshirani2019conformal}. Let $(X_1^0,Y_1^0),\ldots,(X_{m}^0,Y_{m}^0)$ be a proper training set of size $m$, used for fitting the regression function $\hat f_0$. Also let $(X_1,Y_1),\ldots,(X_n,Y_n)$ be the calibration set of size $n$ and let $(X_{n+1},Y_{n+1})$ be the test case. Then weighted split conformal prediction is a special case of Corollary~1
in which $\hat f_0$ is treated as fixed, and e.g.\ if we use the absolute-residual score function, the prediction interval simplifies to
\begin{equation*}
\hat{C}_n(x) = \hat f_0(x) \pm \mathrm{Quantile} \biggl( 1-\alpha; \sum_{i=1}^n p_i^w(x) \delta_{|Y_i-\hat f_0(X_i)|} +
p_{n+1}^w(x) \delta_\infty \biggr),
\end{equation*}
with weights $p_i^w$ as defined in \eqref{eqn:Tibs6}. By Corollary~1, this has coverage at least $1-\alpha$ conditional on the proper training set $(X_1^0,Y_1^0),\ldots,(X_{m}^0,Y_{m}^0)$. Similar results apply for other score functions and for classification problems.

Finally, Remark 3 of \cite{tibshirani2019conformal} notes that the results above still hold if the likelihood ratio weights have an unknown normalization constant, i.e.\ if $w \propto \mathrm{d}\tilde{P}_X/\mathrm{d}P_X$, because this constant cancels out in the final weights in \eqref{eqn:Tibs6}.

\subsection{Unequal probability sampling with replacement}\label{sec:PPS}

Now, consider independently sampling WR from a finite population where each unit has unequal but known sampling probabilities, such as in PPS sampling. Assume no stratification, clustering, or other design constraints.
Our goal is to use this survey sampling design to get prediction intervals that cover $Y$ for most of the $N$ units in the finite population.

Below, we show that this is a special case of the covariate shift setting of \cite{tibshirani2019conformal}. Their distribution $P_X$ is replaced with the sampling probabilities for each unit in the finite population, and $\tilde{P}_X$ is replaced with a uniform distribution over the same population units. The only randomness is in the sampling; all covariates and response variables are fixed in the population, as is usual in design-based inference.

Let $j$ index the universe of population units from $1$ to $N$. Once a sample $S$ is selected, it can be written as a length-$n$ vector of sampled unit IDs, so that the elements of $S$ take values in $\{1,\ldots,N\}$. Often, survey researchers find it easiest to work with the random vector $S$ by rewriting it as a length-$N$ vector of sampling indicators $Z_j = \{1 \mathrm{\ if\ } j \in S \mathrm{,\ and\ } 0 \mathrm{\ otherwise}\}$. However, in the covariate shift setting, it will be easier to work with $S$ directly.
While using $j$ to index the population units, we use $i$ for indexing the elements of $S$. For example, $S_1=4$ would mean that the first unit we sampled was the fourth member of the population.

\begin{lemma}\label{lem:PPSWR}
Assume our training sample is a sample WR of size $n$ where each unit is drawn independently from a finite population of size $N$, with possibly-unequal but nonzero and known sampling probabilities $\pi_j$ for each population ID $j \in 1,\ldots,N$. Also assume our test case is a single observation sampled uniformly at random from the population. Let $S$ be a vector whose first $n$ elements are the population IDs of our $n$ training cases, and whose last element is the population ID for our test case.
\begin{itemize}
  \item Let our finite population consist of $N$ units, each with its own fixed data vector $(X_j,Y_j)$ and sampling probability $\pi_j$ for $j=1,\ldots,N$. The conditional distribution $P_{(X,Y)|S}$ is a deterministic lookup table: Once we have sampled a random $S_i$ for some $i\in \{1,\ldots,n+1\}$, we observe its corresponding covariate and response values $(X_{S_i},Y_{S_i})$. 
  \item Let the training distribution $P_S$ correspond to our complex survey design, which consists of $n$ iid draws from a categorical distribution\footnote{The sampling indicator vector $Z$ is a draw from a multinomial distribution (sampling WR) or a multivariate hypergeometric distribution (sampling WOR), where the $j^{th}$ element counts how often that unit was selected. Sampling from a categorical distribution is identical, but instead of recording counts, we record each of the selected unit IDs individually \citep{tu2014dirichlet}.} whose categories are simply the unit IDs $j \in \{1,\ldots,N\}$, with known probabilities $\pi_j$.
  \item Let the ``covariate-shifted'' test distribution $\tilde{P}_S$ consist of one draw from the same set of unit IDs as the training distribution, but with uniform probability: let $\tilde{\pi}_j = 1/N$ for all $j=1,\ldots,N$.
\end{itemize}
Then this is a special case of the covariate shift setting, with $S$ and $(X,Y)$ playing the roles that $X$ and $Y$ respectively played in \eqref{eqn:Tibs5}:
\begin{align}
(S_i, (X_{S_i}, Y_{S_i})) &\stackrel{iid}{\sim} P = P_S \times P_{(X,Y)|S},\ i=1,\ldots,n, \nonumber \\
(S_{n+1}, (X_{S_{n+1}}, Y_{S_{n+1}})) &\sim \tilde{P} = \tilde{P}_S \times P_{(X,Y)|S},\ \textrm{independently,}
\end{align}
and $w(S_i) \equiv 1/\pi_{S_i}$.
\end{lemma}
\begin{proof}
In the training data sampling design, units are sampled with replacement from a fixed population, and thus they are iid. Even though individual population units have different sampling probabilities, each unit in a training sample is drawn from the same categorical distribution. Also, because each $(X_j,Y_j)$ is paired with a fixed $j$ for all $j=1,\ldots,N$, the conditional distribution $P_{(X,Y)|S}$ is the same for training and test data.
Hence, these training and test distributions match the requirements in \eqref{eqn:Tibs5}.

Furthermore, since $\pi_j$ is nonzero for all population units, we have $w = \mathrm{d}\tilde{P}_S / \mathrm{d}P_S = (1/N) / \pi_{S_i} \propto 1/\pi_{S_i}$. By Remark~3 of \cite{tibshirani2019conformal}, it is safe to ignore the normalization constant and set $w(S_i) = 1/\pi_{S_i}$ directly.
\end{proof}

\setcounter{corollary}{1}
\begin{corollary}\label{cor:PPSWR}
Assume that $n$ training cases and one test case are drawn as described in Lemma~\ref{lem:PPSWR}. Then Lemma~2, the weighted full conformal results from Corollary~1, and the weighted split conformal results from Section~A.3 of \cite{tibshirani2019conformal} hold, with $p_i^w$ defined by using the inverse-probability sampling weights $1/\pi_{S_i}$ for the likelihood ratio weight function $w$ in \eqref{eqn:Tibs6}.
\end{corollary}
\begin{proof}
In this setting, ${P}_S$ and $\tilde{P}_S$ are discrete distributions with identical support, so $\tilde{P}_S$ is absolutely continuous with respect to ${P}_S$. All other conditions of Corollary~1 and Section~A.3 of \cite{tibshirani2019conformal} are met by Lemma~\ref{lem:PPSWR}.
\end{proof}

As suggested by the intuition in Section~\ref{sec:IntuitionWtd}, we carry out conformal inference by replacing the exchangeable eCDF with a survey-weighted eCDF in which the first $n$ observations (the training sample) have their usual inverse-probability sampling weights. The $(n+1)^{th}$ sample (the test case) is assigned the known sampling weight $1/\pi_{S_{n+1}}$ that population unit $S_{n+1}$ would have had under the original sampling design.
The assumption of uniform sampling of test cases simply lets us guarantee marginal coverage across the entire finite population. We do need to know what $\pi_j$ would have been for every unit in the population, which may be reasonable for the organization carrying out the survey but not for end users of the data; Section~\ref{sec:Practical} discusses ways to address this issue.

In a situation with no covariates, we can replace $X$ with a constant in Lemma 3, and apply Lemma 2 to get a prediction interval for $Y$ that is not conditional on $X$, although it may depend on $\pi_{S_{n+1}}$.



\subsection{Unequal probability sampling without replacement}

SRSWOR is exchangeable, so the usual conformal methods apply directly.
But other kinds of sampling WOR are not a special case of the covariate shift setting above, because the training data are no longer independent. While \cite{tibshirani2019conformal} do provide more general results under a relaxed condition they call ``weighted exchangeability,'' it is not immediately clear that this condition can account for sampling WOR.

However, if $n \ll N$, then statistical properties derived under sampling WR are often fairly good approximations to the actual properties under assuming sampling WOR. Our simulations in Section~\ref{sec:Sims}, in which we sample WOR, suggest that this is likely to hold true for conformal methods as well.

\subsection{Cluster sampling}\label{sec:Cluster}

We cannot apply the covariate shift results to cluster sampling, because the data are not independent. Cluster sampling with unequal probabilities will require further research.

However, in the special case where the clusters themselves are sampled by SRS and the ultimate units are sampled by SRS within each cluster, then we can apply the methods of \cite{dunn2022distribution}. Their paper is framed in terms of a more general two-layer hierarchical setting. They do not explicitly consider a finite-population setting, but their assumptions do allow for it (except for some restrictions in their CDF pooling method).

In the framework of \cite{dunn2022distribution}, let $P_1,\ldots,P_k \sim \Pi$ be $k$ random distributions drawn iid from $\Pi$. From each of the sampled distributions $P_\ell$ for $\ell=1,\ldots,k$, we draw $n_\ell$ iid observations $(X_{\ell 1}, Y_{\ell 1}),\ldots,(X_{\ell n_\ell}, Y_{\ell n_\ell})$.

The corresponding setup in survey sampling would be cluster sampling, where our finite population of size $N$ is partitioned into a fixed number $K$ of clusters or Primary Sampling Units (PSUs), and we take a sample of these clusters. SRSWR from a finite population is a special case of iid sampling. Hence, if we take a SRSWR of $k<K$ clusters $P_1,\ldots,P_k$ from the finite population, and then take a SRSWR of $n_\ell$ ultimate units from cluster $\ell$ for each $\ell=1,\ldots,k$, then this is a special case of the setup above, and we can apply most of the results in \cite{dunn2022distribution}. We outline one of their approaches briefly here, and the others in Appendix~\ref{sec:Dunn}, but refer readers to their full paper for details.

Although \cite{dunn2022distribution} state their results in terms of iid sampling, it seems likely that they could be relaxed to exchangeable sampling of the distributions $P_\ell$ as well as exchangeable sampling within each $P_\ell$. If so, these results would also apply to SRSWOR, not just SRSWR. Similarly, we conjecture that some of their results could be extended to the ``weighted exchangeable'' setting of \cite{tibshirani2019conformal}.

\paragraph{Subsampling:} By subsampling, \cite{dunn2022distribution} change the sampling design to become exchangeable. Start with the design above, but then subsample our dataset by choosing one unit at random from each cluster. Then any test case from any new cluster is exchangeable with our subsample. We can treat the $k$ subsampled training cases and the one test case as being generated exchangeably by the process: ``Take a cluster at random, then take one observation at random from that cluster,'' and it is valid to use standard conformal methods.

However, although this process guarantees exact marginal coverage $1-\alpha$ across training sets, it ignores most of the data and leads to wider variability in achieved coverage from training set to training set. An alternative is to carry out repeated subsampling $B$ times and combine the results appropriately across subsamples. \cite{dunn2022distribution} show how to do this in a way that is guaranteed to have coverage of $1-2\alpha$, but in practice tends to achieve coverage close to $1-\alpha$.

\subsection{Stratified sampling}\label{sec:Strata}

Again, we cannot apply the covariate shift results to stratified sampling, because even though strata are independent of each other, the $n$ samples are not independent.

However, we can safely apply the conformal methods from previous subsections within each stratum separately, if within each stratum independently we have used one of the sampling methods with conformal guarantees.
In other words, if the full population is partitioned into $H$ strata, we can treat each stratum $h=1,\ldots,H$ as its own population. To form a prediction interval for a test case from stratum $h$, we apply conformal methods to only the $n_h$ training cases from that stratum. Clearly this will guarantee conditional coverage by stratum. If the same coverage level is used simultaneously across all strata, it will also guarantee marginal coverage. This is an example of ``group-balanced conformal prediction'' or ``object-conditional Mondrian conformal prediction'' \citep{vovk2013conditional, vovk2022algorithmic, angelopoulos2022gentle}.

This stratum-by-stratum approach will lead to a loss of statistical efficiency, since each stratum's conformal quantiles will be estimated using a sample size $n_h<n$. On the other hand,
in some situations, prediction intervals may be more useful if we allow their sizes to vary by stratum than if their size has to be constant across strata. Further, guaranteeing coverage conditional on stratum may be more desirable than only guaranteeing marginal coverage, which could be achieved by overcoverage in some strata at the expense of undercoverage in others.


\subsection{Post-stratification}\label{sec:Poststrat}

Imagine our first $n$ observations were drawn SRSWR, but we wish to post-stratify after data collection, and the population size $N_h$ of each post-stratum is known. We could reweight each sampled observation by the relative stratum sizes: unit-level post-stratification weights are proportional to $N_h/n_h$, where $n_h$ is random rather than fixed in advance \citep{lohr2021sampling}.

We cannot use such weights for conformal inference and retain our exact finite-sample guarantees under the justifications in the present paper, because it would induce dependence between the training set and test case.
However, \cite{fannjiang2022conformal} extend conformal methods to allow for ``feedback covariate shift,'' where the test distribution is allowed to depend on the observed training data, and this may be a promising direction for future work on post-stratified conformal prediction.

In the meantime, we can treat post-stratification as an approximation to estimating the covariate-shift likelihood ratio weights. Although we lose the exact conformal guarantees, using such an approximation would be just as reasonable as the estimation of covariate-shift weights in general. Specifically, imagine we are sampling SRSWR from two different finite populations: a training population of size $M$, and a test population of size $N$. Both populations have the same $P_{(X,Y)|S}$ and the same set of post-strata $1,\ldots,H$, but different (and known) post-stratum sizes $M_h,N_h$ for $h=1,\ldots,H$. In both cases, our sampling design is equivalent to first choosing a post-stratum at random with probability proportional to post-stratum size, then a unit from within that post-stratum at random.

Now, we apply Lemma~\ref{lem:PPSWR}---except that we let $P_S$ and $\tilde{P}_S$ depend on the post-stratum ID $h \in 1,\ldots,H$, not the population unit ID. Then $dP_S/d\tilde{P}_S = (M_h/M)/(N_h/N)$, so $w(S_i) \propto N_{h_i}/M_{h_i}$. So far, we have exact guarantees. If we now assume that we do not actually know the true post-stratum sizes for the training population, we can replace $M_h/M$ with training-sample estimates $n_h/n$ and get post-stratification weights $w(S_i) \propto N_{h_i}/n_{h_i}$. If we further assume that both populations are actually the same, we are now justified in using conformal methods with the standard post-stratification weights. Our guarantees are approximate only because we estimated the ``training'' post-stratum sizes.

Similar types of weighting could also be developed in order to apply conformal inference when we do not we wish to assume that test cases will be sampled uniformly but with some other sampling design. For instance, instead of guaranteeing prediction interval coverage across people, perhaps we want to guarantee coverage across visits to the doctor, and we use a sampling distribution to encode our knowledge of different people's propensities to visit the doctor.

\section{Examples}\label{sec:Examples}

Our R code and knitted RMarkdown output are available at \\ \url{https://github.com/ColbyStatSvyRsch/surveyConformal-paper-code} .

\subsection{Real data}

%
%
%

We have claimed that conformal methods may work better when they account for the sampling design of the data. As a simple demonstration, we turn to an extract of the Medical Expenditure Panel Survey or MEPS \citep{ahrq2017meps}, which is a nationally representative survey about the cost and use of health care among the U.S.\ civilian noninstitutionalized population.

We chose the MEPS because it has already been used as a benchmark dataset in several conformal inference papers, starting with \cite{romano2019conformalized} and followed by others \citep{sesia2020comparison, sesia2021conformal, feldman2021improving, bai2022efficient}. In each of these papers, the authors randomly partition MEPS data into proper training, calibration, and test sets, then report the coverage and length of conformal prediction intervals (PIs) for various models across many such random partitions.
However, none of these papers report accounting for the complex sampling design of MEPS, which includes stratification, clustering, and oversampling of selected subgroups.

At present, we do not attempt a full correction of these earlier analyses of MEPS. We only wish to illustrate that there can be noticeable differences in the conformal PIs depending on whether or not we account for the sampling design, even in a very simple analysis. We use a portion of the public-use dataset for calendar year 2015. We subset to only those respondents who filled out the self-administered questionnaire (SAQ) portion of the survey, and we use the person-level weight variable designed to be used with the SAQ for persons age 18 and older during the interview.

In the poster associated with \cite{romano2019conformalized}, available at
\url{https://github.com/yromano/cqr/blob/master/poster/CQR_Poster.pdf},
the authors explain that they are predicting ``health care utilization, reflecting \# visits to doctor's office / hospital.''
Following their GitHub code, we define a ``utilization'' response variable as the sum of five counts for 2015:
total number of office-based visits reported;
total number of reported visits to hospital outpatient departments;
count of all emergency room visits reported;
total number of nights associated with hospital discharges; and
total number of days where home health care was received from any type of paid or unpaid caregiver.

Unlike the earlier conformal analyses of MEPS, we do take into account the public-use variables for strata, PSUs, and person-level weights. In the 2015 SAQ-eligible subset that we work with, there are 165 strata, and most have 2 or 3 PSUs. First we drop the 4 strata which had no observations in either PSU 1 or PSU 2. Next, for simplicity, we set aside every observation whose PSU is labeled 3 (regardless of stratum) and treat them as our test set. We treat the rest (PSUs 1 and 2) as our overall training set. We split this training data into proper-training and calibration sets under two different approaches. The first approach is a 50/50 SRS split that ignores the survey design. The second approach is to form a random split by PSU within each stratum, so that in each stratum independently, we randomly assign either PSU 1 to proper-training and PSU 2 to calibration or vice versa. These approaches result in proper-training and calibration sets with around 10,000 people each and a test set with 1659 people.

For each split, we fit a linear regression model to the proper-training set to predict utilization from a subset of the covariates used by \cite{romano2019conformalized}: age; sex; indicators for diabetes diagnosis, private insurance coverage, and public insurance coverage; and quantitative summaries of answers to the Physical Component Summary (PCS), the Mental Component Summary (MCS), and the Kessler Index (K6) of non-specific psychological distress. Higher PCS and MCS scores represent better health, while lower K6 scores represent less distress. We calculate PIs for each test case by combining it with the calibration set and finding conformal quantiles.
Developing conformal methods for designs with only one or two PSUs per stratum is still an open problem, not yet addressed by
the methods of Sections~\ref{sec:Cluster} and \ref{sec:Strata}.
Collapsing strata into pseudo-strata could be a reasonable solution if we had subject matter knowledge of the strata, but the public-use MEPS data uses anonymized stratum IDs.
For this reason, our quantiles do not use these methods
to handle the clustering and stratification, but they do apply the survey weights as in Section~\ref{sec:PPS}.


Because we ended up with such large proper-training and calibration sets, relative to the smaller test set, we only saw small differences between the exchangeable and design-based conformal approaches at moderate PI levels. On the other hand, we have enough proper-training and calibration data to estimate 99\% PI levels too, and there we do see substantial differences in the average PI length.
Our main takeaways, based on Tables~\ref{table:meps_covg} and \ref{table:meps_len}:

\begin{table}[t!]
\centering
\begin{tabular}{rllllll}
  \hline
  & Split/fit/conformal & Test set & 80\% PI covg & 90\% PI covg  & 95\% PI covg & 99\% PI covg\\
  \hline
  & SRS & SRS & (0.824, 0.826) & (0.910, 0.911) & (0.951, 0.951) & (0.993, 0.994) \\
  & SRS & Svy-wtd & (0.836, 0.838) & (0.914, 0.915) & (0.950, 0.951) & (0.995, 0.995) \\
  & Svy-wtd & Svy-wtd & (0.829, 0.831) & (0.911, 0.913) & (0.949, 0.950) & (0.992, 0.993) \\
  \hline
\end{tabular}
\caption{Linear regression models' PI coverage, estimated on the MEPS dataset. Coverages reported as approximate 95\% CIs for the average, based on 100 random proper-training/calibration splits at each setting, using the same test set each time. When data splits, proper-training-set model fits, and calibration-set conformal quantiles ignored the survey design, we over-covered (especially for lower PI levels); but when test-set estimates of coverage also ignored the sampling design, they \emph{underestimated} the amount of overcoverage, compared to test-set estimates that did account for the sampling design. However, when splits, fits, and conformal quantiles accounted for the survey design, there was slightly less over-coverage.}
\label{table:meps_covg}
\end{table}

\begin{table}[t!]
\centering
\begin{tabular}{rllllll}
  \hline
  & Split/fit/conformal & Test set & 80\% PI length & 90\% PI length  & 95\% PI length & 99\% PI length\\
  \hline
  & SRS & Either & (28.9, 29.2) & (43.9, 44.3) & (60.2, 60.7) & (250.5, 256.6) \\
  & Svy-wtd & Svy-wtd & (27.3, 27.7) & (40.9, 41.4) & (57.4, 58.0) & (202.0, 211.1) \\
  \hline
\end{tabular}
\caption{Linear regression models' PI lengths, estimated on the MEPS dataset. PI lengths reported as approximate 95\% CIs for the average, based on 100 random proper-training/calibration splits at each setting, using the same test set each time. When data splits, proper-training-set model fits, and calibration-set conformal quantiles ignored the survey design, our PI lengths tended to be slightly larger than when splits, fits, and conformal quantiles did account for the survey design---or much larger when the PI level is very high. In the Table's first row, it does not matter whether or not test-set estimates were survey-weighted, because these PI lengths are constant across test-set cases for a given data split and PI level.}
\label{table:meps_len}
\end{table}

\begin{enumerate}
  \item When we used a conformal pipeline that assumed exchangeability (data splits at random; no weights in model-fitting; no weights in the conformal quantiles on the calibration set), we tended to over-cover. If we also ignored the weights when using the test set to estimate coverage, these calculations under-estimated just how much over-coverage there was. By taking survey-weighted means on the test set, we found slightly higher estimates of coverage. We believe these higher estimates are more appropriate, since the survey-weighted means ought to generalize to the rest of the population better than un-weighted means do.
  \item When we did use design-based methods (design-based splits; design-based and survey-weighted model fits; and survey-weighted conformal quantiles on the calibration set), this reduced our over-coverage a little, according to the survey-weighted means of coverage on the test set. Similarly, it also made our PIs a little narrower (around one to three fewer utilizations/year) for moderate PI levels, and substantially narrower (around forty fewer utilizations/year) for 99\% PIs.
\end{enumerate}

This brief MEPS example demonstrates that the estimated PI coverages and lengths can differ when conformal methods account for the survey design. In the next subsection, we study these effects in more detail, by repeatedly sampling under known sampling designs from a complete finite population. We also use smaller sample sizes, to see more pronounced differences between using vs.\ ignoring the survey design.

\subsection{Simulations}\label{sec:Sims}

For our design-based simulations, we used the Academic Performance Index (API) data \citep{cdoe2018academic} from \texttt{R}'s \texttt{survey} package \citep{lumley2021survey}. The \texttt{apipop} dataset contains information on 37 variables for all 6194 California schools (elementary, middle, or high school) with at least 100 students. The dataset vintage is not documented, but appears to be the 1999-2000 academic year, since the data includes API scores for each school for 1999 and 2000.

We used the \texttt{apipop} dataset as the finite population, and repeatedly took samples (with around $n=200$ ultimate sampling units) using various designs. When we evaluated our results on test sets, we used the entire finite population---including those cases that had already been used to fit models or find conformal quantiles---because this corresponds to the guarantees that our paper makes in Section~\ref{sec:Methods}.
Simulation details:
\begin{itemize}
  \item In all simulations, sampling was done without replacement. Although our results in Section~\ref{sec:Methods} assume sampling with replacement, we conjectured that sampling without replacement would still lead to conformal coverage close to nominal, and we wanted to check this empirically.
  \item All simulations were run 1000 times. All 95\% confidence intervals are calculated as the estimate $\pm$ 2 times the SD over $\sqrt{1000}$.
  \item After dropping the rows with missing values for \texttt{enroll} and \texttt{mobility}, the full ``finite population'' consisted of the 6153 schools without missing values for any variables used in the simulations.
  \item Most simulations used a sample size of $n=200$ schools. However, the cluster samples had $n \approx 200$ schools on average but varied across samples. The regression model simulations took PPS samples of size $m+n=400$, then split the samples at random into a proper training set of $m=200$ and a calibration set of $n=200$.
  \item The response variable was usually \texttt{api00}, the school's API in 2000. The exception is Table~\ref{table:PPS-enroll}, where the response variable was \texttt{enroll}, the same variable used to create the PPS weights. For the non-regression simulations, we simply sought ``unsupervised'' prediction intervals for the marginal distribution of the response variable (with no covariates). For the regression simulations, we found quantiles of $|y - \hat f(x)|$ on the calibration set and sought prediction intervals for the response variable at the covariate values for each unit in the population.
  \item PPS sampling probabilities (if used) were usually proportional to \texttt{enroll}, the number of students enrolled at the school. The exception is parts of Tables~\ref{table:model_covg} and \ref{table:model_len}, where PPS probabilities were proportional to 1 plus the square root of the residuals from the full-population linear regression model, in order to see the effects of over-sampling cases that are hard to fit well. Conformal quantiles for the PPS simulations were calculated as in Section~\ref{sec:PPS}.
  \item Clusters (if used) were based on \texttt{dnum}, the school district number. Cluster samples always took a SRS of 24 school districts. 24 was chosen because it led to an average of $n=198$ schools (close to the $n=200$ used in other sampling designs). For the ``survey-design-aware'' results in Table~\ref{table:clus}, we calculated quantiles using the ``subsampling once'' method as in Section~\ref{sec:Cluster}, while the design-unaware results ignored clustering and calculated quantiles on the whole dataset.
  \item Strata (if used) were based on \texttt{stype}, the school type. Stratified samples always took 100 elementary, 50 middle, and 50 high schools, with an SRS within each school type. For the ``survey-design-aware'' results in Table~\ref{table:strat}, we calculated quantiles separately by stratum as in Section~\ref{sec:Strata}, while the design-unaware results ignored strata and calculated quantiles on the whole dataset.
  \item Linear regression models always predicted \texttt{api00} using a linear combination of \texttt{ell} (the percentage of English language learners), \texttt{meals} (the percentage of students eligible for subsidized meals), and \texttt{mobility} (the percentage of students for whom this is the first year at the school).
\end{itemize}

Our main takeaways:

\begin{enumerate}
  \item Across many settings, using the naive quantile (the $\lceil n\alpha\rceil$ order statistic) instead of the conformal quantile (the $\lceil (n+1)\alpha\rceil$ order statistic) tended to give slight undercoverage. The conformal quantile helped partly to fix this; but in non-SRS settings it was not enough of a fix on its own.
  \item For SRS designs, the conformal quantile lemma worked as advertised. See Table~\ref{table:SRS}.
  \item For PPS designs, ignoring the weights gave slight undercoverage when weights were not highly informative about the response variable. On the other hand, ignoring the weights led to extreme \emph{over}-coverage when weights were highly informative. In both cases, weighted conformal quantiles fixed the problem. See Tables~\ref{table:PPS-api00} and \ref{table:PPS-enroll}.
  \item For clustered designs, as well as for stratified designs, ignoring the design undercovered but accounting for the design (including conformal-quantile padding) did fix it. The one exception was for one of the cluster-design simulations, where the design-based conformal PIs did not reach the target coverage. This might have been due to the small number of clusters, large variation in cluster sizes, and our choice of ``subsampling once'' as the conformal method. See Tables~\ref{table:clus} and \ref{table:strat}.
  \item For simple models and split-conformal inference,
if the weights were not highly informative about model variables or the fit of the model,
then it did not make much difference whether or not the weights were used for quantiles.
But when the weights were informative,
we saw that unweighted conformal quantiles \emph{over}-covered (and PIs were too wide).
Using survey-weights in model-fitting was not enough to fix it,
but weighting the quantiles was. See Tables~\ref{table:model_covg} and \ref{table:model_len}.
\end{enumerate}


\begin{table}[ht!]
\centering
\begin{tabular}{rlll}
  \hline
 & Conformal? & 80\% PI coverage & 90\% PI coverage \\
  \hline
  & no & (0.794, 0.801) & (0.945, 0.949) \\
  & yes & (0.800, 0.807) & (0.949, 0.953) \\
   \hline
\end{tabular}
\caption{SRS. Average PI coverage of \texttt{api00}, at two different PI levels, under 1000 SRS samples of $n=200$ each from API dataset. Coverages reported as 95\% confidence intervals. Naive quantiles undercover, but conformal quantiles achieve target coverage.}
\label{table:SRS}
\end{table}

\begin{table}[ht!]
\centering
\begin{tabular}{rllll}
  \hline
 & Survey-weighted? & Conformal? & 80\% PI coverage & 90\% PI coverage \\
  \hline
  & no & no & (0.752, 0.760) & (0.927, 0.932) \\
  & no & yes & (0.759, 0.767) & (0.935, 0.940) \\
  & yes & no & (0.795, 0.803) & (0.944, 0.949) \\
  & yes & yes & (0.803, 0.811) & (0.953, 0.958) \\
   \hline
\end{tabular}
\caption{Uninformative PPS. Average PI coverage of \texttt{api00}, at two different PI levels, under 1000 PPS samples of $n=200$ each from API dataset where probability $\propto$ \texttt{enroll}. Coverages reported as 95\% confidence intervals. Naive quantiles \emph{under}-cover; conformal quantiles alone or survey-weighting alone do not fix it; but survey-weighted conformal quantiles achieve target coverage.}
\label{table:PPS-api00}
\end{table}

\begin{table}[ht!]
\centering
\begin{tabular}{rllll}
  \hline
 & Survey-weighted? & Conformal? & 80\% PI coverage & 90\% PI coverage \\
  \hline
  & no & no & (0.933, 0.935) & (0.986, 0.987) \\
  & no & yes & (0.934, 0.937) & (0.988, 0.989) \\
  & yes & no & (0.792, 0.798) & (0.946, 0.949) \\
  & yes & yes & (0.796, 0.802) & (0.948, 0.951) \\
   \hline
\end{tabular}
\caption{Informative PPS. Average PI coverage of \texttt{enroll}, at two different PI levels, under 1000 PPS samples of $n=200$ each from API dataset where probability $\propto$ \texttt{enroll}. Coverages reported as 95\% confidence intervals. Naive quantiles \emph{over}-cover; conformal quantiles alone do not fix it, while survey-weighting alone \emph{under}-covers; but survey-weighted conformal quantiles achieve target coverage.}
\label{table:PPS-enroll}
\end{table}

\begin{table}[ht!]
\centering
\begin{tabular}{rllll}
  \hline
 & Survey-design? & Conformal? & 80\% PI coverage & 90\% PI coverage \\
  \hline
  & no & no & (0.785, 0.805) & (0.934, 0.943) \\
  & no & yes & (0.791, 0.811) & (0.940, 0.949) \\
  & yes & no & (0.799, 0.817) & (0.916, 0.930) \\
  & yes & yes & (0.799, 0.817) & (0.959, 0.968) \\
   \hline
\end{tabular}
\caption{Clustering. Average PI coverage of \texttt{api00}, at two different PI levels, under 1000 clustered samples of 24 clusters ($n\approx 200$) each from API dataset. Coverages reported as 95\% confidence intervals. Due to high variability in cluster sizes, these 95\% CIs are wider than in previous tables, but overall trend is generally similar to other tables: Naive quantiles appear likely to \emph{under}-cover; conformal quantiles alone or survey-design-aware analyses alone do not necessarily fix it; but survey-design-aware conformal quantiles achieve target coverage.}
\label{table:clus}
\end{table}

\begin{table}[ht!]
\centering
\begin{tabular}{rllll}
  \hline
 & Survey-design? & Conformal? & 80\% PI coverage & 90\% PI coverage \\
  \hline
  & no & no & (0.769, 0.777) & (0.933, 0.938) \\
  & no & yes & (0.775, 0.783) & (0.940, 0.944) \\
  & yes & no & (0.787, 0.795) & (0.939, 0.944) \\
  & yes & yes & (0.800, 0.808) & (0.953, 0.956) \\
   \hline
\end{tabular}
\caption{Stratification. Average PI coverage of \texttt{api00}, at two different PI levels, under 1000 stratified samples of $n=200$ each from API dataset. Coverages reported as 95\% confidence intervals. Naive quantiles \emph{under}-cover; conformal quantiles alone or survey-design-aware analyses alone do not fix it; but survey-design-aware conformal quantiles achieve target coverage.}
\label{table:strat}
\end{table}

\begin{table}[ht!]
\centering
\begin{tabular}{rlllll}
  \hline
  & PPS probs & Svy-wtd conformal? & Svy-wtd regression? & 80\% PI covg & 90\% PI covg \\
  \hline
  & \texttt{enroll} & no & no & (0.807, 0.810) & (0.955, 0.957) \\
  & \texttt{enroll} & yes & no & (0.803, 0.807) & (0.954, 0.956) \\
  \hline
  & residuals & no & no & (0.874, 0.876) & (0.973, 0.974) \\
  & residuals & no & yes & (0.875, 0.878) & (0.973, 0.974) \\
  & residuals & yes & no & (0.798, 0.802) & (0.950, 0.951) \\
  & residuals & yes & yes & (0.798, 0.801) & (0.950, 0.951) \\
  \hline
\end{tabular}
\caption{Linear regression models' PI coverage. Average PI coverage of \texttt{api00}, at two different PI levels, under 1000 PPS samples of $n=200$ each from API dataset. Coverages reported as 95\% confidence intervals. For weights proportional to \texttt{enroll} (uninformative for the regression), it makes little difference whether or not we weight the conformal quantiles. For informative weights proportional to full-pop residuals, un-weighted conformal quantiles \emph{over}-cover, and this is fixed by survey-weighted conformal quantiles; but it makes little difference whether or not we fit survey-weighted regression models.}
\label{table:model_covg}
\end{table}

\begin{table}[ht!]
\centering
\begin{tabular}{rlllll}
  \hline
  & PPS probs & Svy-wtd conformal? & Svy-wtd regression? & 80\% PI length & 90\% PI length \\
  \hline
  & \texttt{enroll} & no & no & (192.6, 194.1) & (294.5, 297.3) \\
  & \texttt{enroll} & yes & no & (190.8, 192.7) & (293.7, 297.2) \\
  \hline
  & residuals & no & no & (214.1, 215.6) & (328.2, 331.5) \\
  & residuals & no & yes & (211.8, 213.3) & (331.4, 334.8) \\
  & residuals & yes & no & (179.4, 180.8) & (285.4, 287.7) \\
  & residuals & yes & yes & (175.1, 176.3) & (286.1, 288.3) \\
  \hline
\end{tabular}
\caption{Linear regression models' PI lengths. Average length of PIs for \texttt{api00}, at two different PI levels, under 1000 PPS samples of $n=200$ each from API dataset. Lengths reported as 95\% confidence intervals. For weights proportional to \texttt{enroll} (uninformative for the regression), it makes little difference whether or not we weight the conformal quantiles. For informative weights proportional to full-pop residuals, PIs from un-weighted conformal quantiles are much wider than PIs from survey-weighted conformal quantiles; but it makes little difference whether or not we fit survey-weighted regression models.}
\label{table:model_len}
\end{table}

Overall, the survey-conformal quantiles
we proposed mathematically in Section~\ref{sec:Methods}
also appear to work empirically.
Data analysts will likely get coverage closer to nominal when they account for the weights or other survey design features.

\section{Extensions}\label{sec:Extensions}

We discuss several practical considerations: What if the sampling design does not quite match the situations above? What if the sampling probabilities are not all known? We also suggest other possible use cases for conformal inference in survey methodology.

\subsection{Practical considerations}\label{sec:Practical}

\paragraph{Weighting adjustments:} Most surveys are not released with inverse-probability sampling weights alone. The final survey weights have been adjusted for nonresponse, post-stratification, and other considerations.
\cite{tibshirani2019conformal} found that their conformal methods still maintained coverage close to nominal even when they only estimated the likelihood ratio weights, instead of using true likelihood ratios. We anticipate similar outcomes for conformal methods that use adjusted sampling weights instead of the true inverse-probability sampling weights.

\paragraph{Distribution shift:} Surveys are typically assumed to be sampled from a well-defined finite population, in a specific time and place. We may not be guaranteed coverage if we make predictions for units at future times or from distinct populations.
If model-based (rather than design-based) inference makes more sense for the application at hand, we may still be able to estimate the covariate shift likelihood ratio and apply the conformal techniques of \cite{tibshirani2019conformal}; it will simply no longer be strictly design-based inference, though a joint superpopulation / design-based framework may be a fruitful topic for future work \citep{isaki1982survey,rubin2005two,han2021complex}.

%

\paragraph{Unknown sampling probabilities for test cases:} Our conformal methods of Section~\ref{sec:PPS} require the sampling probabilities for each test case. However, in practice the sampling probabilities are not typically known for every population unit. Even if they are known internally within the survey organization, public release of the entire population's sampling probabilities can increase the risk of unit re-identification and privacy breaches.

Instead, the survey organization could report a set of population categories for which the sampling probabilities are approximately equal within each category, along with their approximate probabilities. There may be categories that are broad enough to minimize privacy risks, but fine enough to approximate the real sampling probabilities well.
Or instead of discrete categories, the survey organization could report a kind of generalized variance function or GVF \citep{wolter2007generalized}, but used to estimate each population unit's sampling probabilities rather than variances, based on covariates available for each unit. 
  This may be especially reasonable for PPS.

Alternatively, the survey organization could release only the value of the single largest inverse-probability sampling weight for the whole population. Then, using that one weight for every test case during conformal prediction would be strictly conservative. Or as an approximation, users could use the largest sampling weight in the public-use dataset.
Users could also try a range of plausible weights for a desired test case, based on the weights of similar in-sample units, and report a sensitivity analysis.

Finally, if a user is interested in predictions for a population unit with the same vector of covariates $X$ as one of the sampled cases in the dataset, they could simply use that sampled case's survey weight.

\subsection{Other uses for conformal methods in survey methodology}

National statistical offices, polling agencies, and any others who collect and pre-process survey data may find their own use cases for conformal methods:

\begin{itemize}
  \item Build a response propensity model based on internal metadata from past surveys. For the next survey, build conformal prediction sets for each sampled unit's most likely mode of response (or ultimate status of nonresponse). Use these sets to help choose the mode of initial contact for each sampled unit.
  \item As part of automated quality control checks, use conformal prediction intervals or sets for each variable to flag potential outliers for followup \citep{bates2023testing}.
  \item To impute for item nonresponse, or to generate synthetic microdata, draw at random from a conformal prediction interval or set for that variable.
\end{itemize}

\section{Conclusion}\label{sec:Conclusion}

There is growing interest in extending machine learning (ML) methods to complex survey data, as well as a distinct body of work around developing prediction intervals for new predictive methods such as those arising from the ML community. While this ongoing work is valuable, design-based conformal inference provides an alternative approach for both needs: We can apply a novel predictive algorithm to complex survey data (even if that algorithm has not been specifically adapted to account for the survey design yet), and automatically get conformal intervals compatible with the fitted prediction function (even if native prediction-interval methods have not been developed for that algorithm yet), and still manage to provide exact, finite-sample, design-based coverage guarantees.


Of course, design-based conformal methods are not a panacea.
First, for methods that do have well-understood design-based adaptations, the design-based version fitted to complex-survey-design training data is likely to be a better predictive model for the conditional mean and consequently to have narrower conformal prediction intervals than when that method ignores the sampling design. As a simple example, in certain cases a survey-weighted linear model may generalize better than an unweighted linear model fit to the same data, so the survey-weighted model's smaller residuals will lead to narrower conformal prediction intervals than for the unweighted model. In this sense, it is still important to develop survey-weighted equivalents of novel ML algorithms.

Second, for models that do have native prediction intervals, if we can justifiably trust that the underlying model assumptions are met, then we may be able to get narrower or less-variable prediction intervals natively than from conformal inference. And in some cases, native methods can give us conditional prediction intervals rather than marginal ones.
In this sense, it is still important to develop native prediction intervals for specific models.

But when either of these conditions is not yet met, conformal inference is a practical and assumption-lean way to fill in the gap. We also note that \cite{lei2018distribution} found in simulations that conventional methods are so noisy for high-dimensional regression that conformal prediction intervals can actually be narrower than native ones.

We have also discussed several pragmatic limitations to conformal inference with survey data: Sampling weights are not truly known in advance, nor can they be reported for every population unit.
Our suggested solutions are only a starting point.

Still, we hope that this paper spurs interest in conformal inference among survey statisticians. Survey data analysts might find these methods directly applicable. 
Furthermore, research into conformal methods from the design-based perspective could bring new insights back to the wider statistics/ML community of conformal inference researchers.

As one example, \cite{vovk2013conditional} 
shows that the distribution of achieved coverage levels has a Beta distribution across calibration sets when data are exchangeable. In order to gauge how much variability to expect from this Beta distribution, \cite{angelopoulos2022gentle} suggest one estimate of the effective sample size $n_\mathrm{eff}$ for weighted conformal methods, but their estimate
is not appropriate for all sampling designs nor for all estimators. Survey statisticians may be able to suggest better estimates of $n_\mathrm{eff}$ or recommend other ways to study and control the variability in achieved coverage for non-exchangeable data.


Other open problems include conformal methods for unequal-probability sampling WOR; unequal-probability cluster sampling; strata with few PSUs; combining strata rather than analyzing each stratum separately; exact coverage guarantees for post-stratification; panel survey designs; or joint superpopulation / design-based frameworks. We encourage survey researchers to contribute to conformal methodology and find new areas of application for these methods.

\section*{Acknowledgements}

The author thanks the Editor, the Associate Editor, and the anonymous referees for their helpful suggestions which improved this manuscript. The author is also grateful to Ryan Tibshirani, Robin Dunn, Benjamin LeRoy, and Evan Randles for their early feedback on several of the ideas presented here.

\bibliographystyle{apalike}
\bibliography{conformal-refs}

\newpage

\appendix

\section{Adaptive prediction regions}\label{sec:Adaptive}

In the simple approach to conformal prediction for regression described in our paper, we have used a constant-width prediction band everywhere, which might be unrealistic for many scenarios. If the true conditional distribution of $Y|X$ is heteroscedastic for instance, marginal coverage will still be correct, but it will be achieved by overcovering at some regions of $X$ and undercovering at others. For ``adaptive'' alternatives, where the PI is wider at regions of $X$ with more variability in $Y$, see recent work on locally-weighted conformal inference \citep{lei2018distribution} and conformalized quantile regression \citep{romano2019conformalized}.

For classification problems, even the simple approach of Section~\ref{sec:Class} produces adaptive prediction sets: for hard test cases where $\hat f$ is uncertain about the right class, prediction sets will be larger than for easy test cases where $\hat f$ confidently assigns most of the probability to one class. See also \cite{romano2020classification} and \cite{angelopoulos2022gentle}.


\paragraph{Varying PI widths due to survey weighting:}
Note that in the survey-weighted setting, we will automatically get slightly different prediction interval widths at different test cases, because the survey weights can differ for each unit being predicted. However, the effect of survey weighting is to widen slightly the prediction intervals for regions of $X$ with a smaller effective sample size. This is distinct from adaptivity to regions with more variability in $Y$.

Loosely, a sampled unit with a larger sampling weight represents more population units, so we are more uncertain about units ``like'' this one. When we apply Lemma~\ref{lem:PPSWR} to a test case with a larger weight, its survey-weighted eCDF is pushed down farther than if it had a small weight; so the nonconformity score quantile is estimated farther to the right; so the conformal prediction interval is wider for units with larger weights, all else being equal.

On the other hand, consider optimal allocation designs, in which strata with higher variance of $Y$ are oversampled. When we construct conformal intervals separately within each stratum, the high-variance strata will have their eCDFs padded by a smaller $1/n_h$ than low-variance strata do, so the quantile adjustment is less conservative. But this should be more than offset by the fact that high-variance strata also have wider spread in $Y$, so that ultimately their prediction intervals will end up wider than those of low-variance strata.


\section{Other conformal methods for cluster samples}\label{sec:Dunn}

Beyond Section~\ref{sec:Cluster}, we briefly note how \cite{dunn2022distribution}'s other methods relate to cluster designs in survey sampling.

\paragraph{Prediction for an observed cluster:} If we only need a prediction interval for new observations from one of the clusters we already sampled, we can simply apply standard conformal methods by only using that cluster's data, which will be exchangeable.

\paragraph{Double conformal:} For unsupervised prediction (where we want a prediction interval for $Y$ without conditioning on covariates $X$) for a new cluster, one could create conformal prediction intervals separately within each cluster, then combine their endpoints appropriately into a single interval. This lines up with the design-based spirit: construct valid estimates within each cluster, then combine them sensibly across clusters. \cite{dunn2022distribution} derive such a method that is guaranteed to have coverage at least $1-\alpha$. However, in simulations it overcovers, with coverage of nearly 1 and wider intervals than other approaches.

\paragraph{Pooling CDFs:} We could first construct eCDFs within each cluster and average them together into one pooled eCDF. Then we could apply standard conformal methods using this pooled eCDF. \cite{dunn2022distribution} only prove that this is asymptotically valid, requiring a continuous distribution for $Y$ as well as a growing number of sampled clusters $k \rightarrow \infty$. This is not possible in the traditional design-based setting of a fixed finite population, although we could construct a superpopulation model and a sequence of growing finite populations that satisfies their requirements.

In simulations, \cite{dunn2022distribution} find that CDF pooling tends to have coverage closest to nominal as well as shortest length of prediction intervals across most settings. But if we do not wish to rely on asymptotic arguments and continuous $Y$ data, repeated subsampling appears to work better than single subsampling or the double conformal method.

\end{document}